\documentclass[a4paper,11pt, draft]{article}
\usepackage[utf8x]{inputenc}
\usepackage{amstext,amsmath,amsthm,amssymb,amscd}
\usepackage{xspace}
\usepackage{fullpage}
\usepackage[linesnumbered, vlined, ruled]{algorithm2e}
\usepackage{authblk}
\newtheorem{theorem}{Theorem}
\newtheorem{definition}{Definition}
\newtheorem{lemma}{Lemma}

\newtheorem{observation}{Observation}

\newcommand{\probname}[1]{\textsc{\lowercase{#1}}}
\newcommand{\problem}[1]{\probname{#1}\xspace}

\newcommand{\VCk}{\probname{Vertex Cover($k$)}\xspace}
\newcommand{\EDSk}{\probname{Edge Dominating Set($k$)}\xspace}
\newcommand{\notdHS}{\probname{not} $d$-\probname{Hitting Set}\xspace}
\newcommand{\dHSk}{$d$-\probname{Hitting Set($k$)}\xspace}

\newcommand{\dSP}{$d$-\probname{Set Packing}\xspace}
\newcommand{\dSPk}{$d$-\probname{Set Packing($k$)}\xspace}
\newcommand{\dHFVDk}{$\mathcal{H}$-\probname{Free Vertex Deletion($k$)}\xspace}
\newcommand{\dHPk}{$H$-\probname{Packing($k$)}\xspace}
\newcommand{\F}{\ensuremath{\mathcal{F}}\xspace}

\newcommand{\Pee}{\ensuremath{\mathcal{P}}\xspace}
\newcommand{\N}{\mathbb{N}}
\newcommand{\Oh}{\mathcal{O}}
\newcommand{\yes}{\textsc{yes}\xspace}
\newcommand{\no}{\textsc{no}\xspace}
\newcommand{\coNP}{\ensuremath{\mathsf{coNP}}\xspace}
\newcommand{\NP}{\ensuremath{\mathsf{NP}}\xspace}
\newcommand{\containment}{\ensuremath{\mathsf{NP\subseteq coNP/poly}}\xspace}
\newcommand{\Rvc}{\ensuremath{R_{\textsc{vc}}}\xspace}
\newcommand{\Reds}{\ensuremath{R_{\textsc{eds}}}\xspace}

\newcommand{\introduceparameterizedproblem}[4]{
\noindent
\fbox{
\parbox{0.97\textwidth}{
	#1\hfill \textbf{Parameter:} #3\\
	\textbf{Input:} #2\\	
	\textbf{Question:} #4}}
\vspace{0.3cm}	
}

\title{A shortcut to (sun)flowers: Kernels in logarithmic space or linear time\thanks{Supported by the Emmy Noether-program of the German Research Foundation (DFG), research project PREMOD (KR 4286/1).}}
\author{Stefan Fafianie}
\author{Stefan Kratsch}
\affil{University of Bonn, Germany, \texttt{$\{$fafianie$,$kratsch$\}$@cs.uni-bonn.de}}

\begin{document}

\maketitle

\begin{abstract}
 We investigate whether kernelization results can be obtained if we restrict kernelization algorithms to run in logarithmic space. This restriction for kernelization is motivated by the question of what results are attainable for preprocessing via simple and/or local reduction rules. We find kernelizations for \dHSk, \dSPk, \EDSk and a number of hitting and packing problems in graphs, each running in logspace. Additionally, we return to the question of linear-time kernelization. For \dHSk a linear-time kernelization was given by van Bevern [Algorithmica (2014)]. We give a simpler procedure and save a large constant factor in the size bound. Furthermore, we show that we can obtain a linear-time kernel for \dSPk as well.
\end{abstract}

\section{Introduction}
The notion of \emph{kernelization} from parameterized complexity offers a framework in which it is possible to establish rigorous upper and lower bounds on the performance of polynomial-time preprocessing for \NP-hard problems. Efficient preprocessing is appealing because one hopes to simplify and shrink input instances before running an exact exponential-time algorithm, approximation algorithm, or heuristic. A well-known example is that given an instance $(G,k)$, asking whether graph $G$ has a vertex cover of size at most $k$, we can efficiently compute an equivalent instance $(G',k')$ where $k'\leq k$ and $G'$ has at most $2k$ vertices \cite{chen2001vertex}. On the other hand, the output instance could still have $\Omega(k^2)$ edges and a result of Dell and van Melkebeek~\cite{dell2010satisfiability} indicates that this cannot be avoided unless \containment (and the polynomial hierarchy collapses). Many intricate techniques have been developed for the field of kernelization and some other variants have been considered. For example, the more relaxed notion of Turing kernelization asks whether a problem can be solved by a polynomial-time algorithm that is allowed to query an oracle for answers to instances of small size \cite{lokshtanov2009new}.

In this work we take a more restrictive view. When considering reduction rules for \NP-hard problems that a human would come up with quickly, these would often be very simple and probably aimed at local structures in the input. Thus, the matching theoretical question would be whether we can also achieve nice kernelization results when restricted to ``simple reduction rules.'' This is of course a very vague statement and largely a matter of opinion. For local reduction rules this seems much easier: If we restrict a kernelization to running in logarithmic space, then we can no longer perform ``complicated'' computations like, for example, running a linear program or even just finding a maximal matching in a graph. Indeed, for an instance $x$, the typical use of $\log |x|$ bits would rather be to store a counter with values up to $|x|^{\Oh(1)}$ or to remember a pointer to some position in $x$.

The main focus of our work is to show that a bunch of classic kernelization results can also be made to work in logarithmic space. To the best of our knowledge such a kernelization was previously only known for \VCk \cite{cai1997advice}. Concretely, we show that \dHSk, \dSPk, and \EDSk as well as a couple of implicit hitting set and set packing type problems on graphs admit polynomial kernels that can be computed in logarithmic space.
The astute reader will instantly suspect that the well-known sunflower lemma will be behind this, but---being a bit fastidious---this is only partially true.

It is well-known that so-called \emph{sunflowers} are very useful for kernelization (they can be used to obtain polynomial kernels for, e.g., \dHSk~\cite{flum2006parameterized} and \dSPk~\cite{dell2012kernelization}). A $k$-sunflower is a collection of $k$ sets $F_1,\ldots,F_k$ such that the pairwise intersection of any two sets is the same set $C$; called the core. The sets $F_1\setminus C,\ldots,F_k\setminus C$ are therefore pairwise disjoint. When seeking a $k$-hitting set $S$, the presence of a $(k+1)$-sunflower implies that $S$ must intersect the core, or else fail to hit at least one set $F_i$. The Sunflower Lemma of Erd\H{o}s and Rado implies that any family with more than $d!k^d$ sets, each of size $d$, must contain a $(k+1)$-sunflower which can be efficiently found. Thus, so long as the instance is large enough, we will find a core $C$ that can be safely added as a new constraint, and the sets $F_i$ containing $C$ may be discarded.

Crucially, the only point of the disjoint sets $F_1\setminus C,\ldots,F_k\setminus C$ is to certify that we need at least $k$ elements to hit all sets $F_1,\ldots,F_k$, assuming we refuse to pick an element of $C$. What if we forgo the disjointness requirement and only request that not picking an element of $C$ incurs a hitting cost of at least $k$ (or at least $k+1$ for the above illustration)? It turns out that the corresponding structure is well-known under the name of a \emph{flower}: A set $\F$ is a $k$-flower with core $C$ if the collection $\{F\setminus C \colon F\in\F, F\supseteq C\}$ has minimum hitting set size at least $k$. Despite the seemingly complicated requirement, H\aa{}stad et al.~\cite{haastad1995top} showed that any family with more than $k^d$ sets must contain a $(k+1)$-flower. Thus, by replacing sunflowers with flowers, we can save the extra $d!$ factor in quite a few kernelizations with likely no increase to the running time. In order to meet the space requirements for our logspace kernelizations, we avoid explicitly finding flowers and instead use careful counting arguments to ensure that a $(k+1)$-flower with core $C \subseteq F$ exists when we discard a set $F$.

Finally, we also return to the question of linear-time kernelization that was previously studied in, e.g., \cite{niedermeier2003efficient, van2014towards}. Using flowers instead of sunflowers we can improve a linear-time kernelization for \dHSk by van Bevern~\cite{van2014towards} from $d!\cdot d^{d+1} \cdot (k+1)^d$ to just $(k+1)^d$ sets (we also save the $d^{d+1}$ factor because of the indirect way in which we use flowers). Similarly, we have a linear-time kernelization for \dSPk to $(d(k-1)+1)^d$ sets.
We note that for linear-time kernelization the extra applications for hitting set and set packing type problems do not necessarily follow: In logarithmic space we can, for example, find all triangles in a graph and thus kernelize \problem{Triangle-free Vertex Deletion($k$)} and \problem{Triangle Packing($k$)}. In linear time we will typically have no algorithm available that can extract the constraints respectively the feasible sets for the packing that are needed to apply a \dHSk or \dSPk kernelization.

We remark that the kernelizations for \dHSk and \dSPk via representative sets (cf.~\cite{kratsch2014recent}) give more savings in the kernel size. For \dHSk this approach yields a kernel with at most $\binom{k+d}{d} = \frac{(k+d)!}{d!+k!} = \frac{1}{d!}(k+1) \cdot \ldots \cdot (k+d) > \frac{k^d}{d!}$ sets, thus saving at most another $d!$ factor. It is however unclear if this approach can be made to work in logarithmic space or linear time. Applying the current fastest algorithm for computing a representative set due to Fomin et al.~\cite{fomin2014efficient} gives us a running time of $\Oh(\binom{k+d}{d}|\F|d^\omega + |\F|\binom{k+d}{k}^{\omega-1})$ where $\omega$ is the matrix multiplication exponent.

\paragraph{Organization}
We will start with preliminaries in Section~\ref{sec:prelim} and give a formal introduction on (sun)flowers in Section~\ref{sec:flowers}. We present our logspace kernelization results for \dHSk, \dSPk, and \EDSk in Sections~\ref{sec:logdhsk}, \ref{sec:logdspk}, and~\ref{sec:logedsk} respectively. In Section~\ref{sec:loggraph} we describe how our logspace kernels for packing and hitting sets can be used in order to obtain logspace kernelizations for implicit hitting and packing problems on graphs. We show how our techniques can be used in conjunction with a data-structure and subroutine by van Bevern~\cite{van2014towards} in order to obtain a smaller linear-time kernel for \dHSk in Section~\ref{sec:lindhsk}. This also extends to a linear-time kernel for \dSPk which we give in Section~\ref{sec:lindspk}.
Concluding remarks are given in Section~\ref{sec:conc}.

\section{Preliminaries} \label{sec:prelim}
 
\paragraph{Set families and graphs.}  We use standard notation from graph theory and set theory. Let $U$ be a finite set, let $\F$ be a family of subsets of $U$, and let $S \subseteq U$. We say that $S$ \emph{hits} a set $F \in \F$ if $S \cap F \neq \emptyset$. In slight abuse of notation we also say that $S$ \emph{hits} $\F$ if for every $F \in \F$ it holds that $S$ hits $F$. More formally, $S$ is a \emph{hitting set} (or \emph{blocking set}) for $\F$ if for every $F \in \F$ it holds that $S \cap F \neq \emptyset$. If $|S| \leq k$, then $S$ is a $k$-hitting set. A family $\Pee \subseteq \F$ is a \emph{packing} if the sets in $\Pee$ are pairwise disjoint; if  $|\Pee| = k$, then $\Pee$ is called a $k$-packing. In the context of instances $(U, \F, k$) for \dHSk or \dSPk we let $n = |U|$ and $m = |\F|$. Similarly, for problems on graphs $G = (V,E)$ we let $n = |V|$ and $m = |E|$. A \emph{restriction} $\F_C$ of a family $\F$ onto a set $C$ is the family $\{F \setminus C: F \in \F, F \supseteq C\}$, i.e., it is obtained by only taking sets in $\F$ that are a superset of $C$ and removing $C$ from these sets.
 
\paragraph{Parameterized complexity.} A \emph{parameterized problem} is a language $Q \subseteq \Sigma^* \times \N$; the second component of instances $(x, k) \in \Sigma^* \times \N$ is called the \emph{parameter}. A parameterized problem $Q \subseteq \Sigma^* \times \N$ is \emph{fixed-parameter tractable} if there is an algorithm that, on input $(x, k) \in \Sigma^* \times \N$, correctly decides if $(x,k) \in Q$ and runs in time $\Oh(f(k)|x|^c)$ for some constant $c$ and any computable function $f$. A \emph{kernelization algorithm} (or \emph{kernel}) for a parameterized problem $Q \subseteq \Sigma^* \times \N$ is an algorithm that, on input $(x, k) \in \Sigma^* \times \N$, outputs in time $(|x|+k)^{\Oh(1)}$ an equivalent instance $(x', k')$ with $|x'| + k' \leq g(k)$ for some computable function $g \colon \N \rightarrow \N$ such that $(x, k) \in Q \Leftrightarrow (x', k') \in Q$. Here $g$ is called the \emph{size} of the kernel; a \emph{polynomial kernel} is a kernel with polynomial size.
 
\section{Sunflowers and flowers} \label{sec:flowers}
The notion of a \emph{sunflower} has played a significant role in obtaining polynomial kernels for the \dHSk and \dSPk problems. We start with a formal definition.
 
\begin{definition} \label{def:sunflower}
 A \emph{sunflower} with $l$ \emph{petals} and \emph{core} $C$ is a family $\F = \{F_1, \ldots, F_l\}$ such that each $F_i \setminus C$ is non-empty and $F_i \cap F_j = C$ for all $i \neq j$.
\end{definition}
 
The prominent \emph{sunflower lemma} by Erd\H{o}s and Rado states that we are guaranteed to find a sunflower with sufficiently many petals in a $d$-uniform set family if this family is large enough.

\begin{lemma}[Erd\H{o}s and Rado \cite{erdos1960intersection}] \label{lem:sunflower}
 Let $\F$ be a family of sets each of cardinality $d$. If $|\F| > d!(l-1)^d$, then $\F$ contains a sunflower with $l$ petals.
\end{lemma}
 
This lemma can be made algorithmic such that we can find a sunflower with $l$ petals in a set family $\F$ in time $\Oh(|\F|)$ if $|\F| > d!(l-1)^d$. Flum and Grohe~\cite{flum2006parameterized} apply this result to obtain a polynomial kernel for \dHSk by repeatedly finding a sunflower with $k+1$ petals and replacing it by its core $C$. This operation preserves the status of the \dHSk problem since any $k$-hitting set $S$ for $\F$ must hit all sets in the sunflower. Because the sunflower has at least $k+1$ petals, $S$ must contain an element of $C$. Alternatively, as used for example by Kratsch~\cite{kratsch2012polynomial}, one can instead look for sunflowers with at least $k+2$ petals and discard sets such that $k+1$ petals are preserved. The presence of these $k+1$ petals in the reduced instance still forces any $k$-hitting set $S$ to hit the core $C$; thus $S$ must hit the discarded sets as well. This has the advantage that, besides preserving all minimal solutions, a subset of the family given in the input is returned in the reduced instance. Kernels adhering to these properties preserve a lot of structural information and are formalized as being \emph{expressive} by van Bevern \cite{van2014towards}. 

Fellows et al.~\cite{fellows2004faster} give a polynomial kernel for \dSPk. Dell and Marx~\cite{dell2012kernelization} provide a self-contained proof for this result that uses sunflowers. Here the crucial observation is that any $k$-packing of sets of size $d$ can intersect with at most $dk$ petals of a sunflower if it avoids intersection with the core (the same argument is implicit in a kernelization for problems from $\mathsf{MAX}$ \NP in \cite{kratsch2012polynomial}). Each of the described kernelization algorithms returns instances of size $\Oh(k^d)$. However, as a consequence of using the sunflower lemma, there is a hidden $d!$ multiplicative factor in these size bounds. We avoid this by considering a relaxed form of sunflower, known as \emph{flower} (cf. Jukna \cite{jukna2011extremal}) instead.

\begin{definition} \label{def:flower}
 An $l$-\emph{flower} with \emph{core} $C$ is a family $\F$ such that any blocking set for the restriction $\F_C$ of sets in $\F$ onto $C$ contains at least $l$ elements.
\end{definition}
 
Note that every sunflower with $l$ petals is also an $l$-flower but not vice-versa. From Definition~\ref{def:flower} it follows that the relaxed condition for set disjointness outside of $C$ is still enough to force a $k$-hitting set $S$ to contain an element of $C$ if there is a $(k+1)$-flower with core $C$. Similar to Lemma~\ref{lem:sunflower}, H\aa{}stad et al.~\cite{haastad1995top} give an upper bound on the size of a set family that must contain an $l$-flower. The next lemma is a restatement of this result. We give a self-contained proof following Jukna's book.

\begin{lemma}[{cf. Jukna \cite[Lemma~7.3]{jukna2011extremal}}] \label{lem:flower}
 Let $\F$ be a family of sets each of cardinality $d$. If $|\F| > (l-1)^d$, then $\F$ contains an $l$-flower. 
\end{lemma}

\begin{proof}
 We prove the lemma for any $l$ by induction over $1 \leq d' \leq d$. If $d' = 1$, then the lemma obviously holds since any $l$ sets in $\F$ are pairwise disjoint and even form a sunflower with core $C = \emptyset$. Let us assume that it holds for sets of size $d'-1$; we prove that it holds for sets of size $d'$ by contradiction. Suppose that there is no $l$-flower in $\F$ while $|\F| > (l-1)^{d'}$. Let $X$ be a minimal blocking set for $\F$. We have that $|X| \leq l-1$, otherwise $\F$ itself is an $l$-flower with core $C = \emptyset$. Since $X$ intersects with every set in $\F$ and $|\F| > (l-1)^{d'}$, there must be some element $x \in X$ that intersects more than $\frac{(l-1)^{d'}}{l-1} = (l-1)^{d'-1}$ sets. Therefore, the restriction $\F_x$ contains more than $(l-1)^{d'-1}$ sets, each of size $d'-1$. By the induction hypothesis we find that $\F_x$ contains an $l$-flower with core $C'$ and obtain an $l$-flower in $\F$ with core $C = C' \cup \{x\}$.
\end{proof}

The proof for Lemma~\ref{lem:flower} implies that we can find a flower in $\Oh(|\F|)$ time if $|\F| > (l-1)^d$ by recursion: Let $\F$ be a family of sets of size $d$ in which we currently want to find an $l$-flower. Pick an element $x$ such that $\F_x$ has more than $(l-1)^{d-1}$ sets; then find a flower in $\F_x$ and add $x$ to its core. If no such $x$ exists, then return $\F$ instead since any set of size $l-1$ intersects with at most $(l-1) \cdot (l-1)^{d-1} = (l-1)^d < |\F|$ sets in $\F$, i.e., a blocking set for $\F$ requires at least $l$ elements. However, in order to obtain our logspace and linear-time kernels we avoid explicitly finding a flower. Instead, by careful counting, we guarantee that a flower must exist with some fixed core $C$ if two conditions are met. To this end we use Lemma~\ref{lem:flowercounting}. Note that we no longer assume $\F$ to be $d$-uniform but instead only require that any set in $\F$ is of size at most $d$, similar to the families that we consider in instances of \dHSk and \dSPk. If the required conditions hold, we find that a family $\F$ either contains an $l$-flower with core $C$ or the set $C$ itself. Thus, any hitting set of size at most $l-1$ for $\F$ must contain an element of $C$. For our \dHSk kernels we use the lemma with $l = k+1$.

\begin{lemma} \label{lem:flowercounting}
 For a finite set~$U$, constant~$d$, and a set~$C \in \binom{U}{< d}$, let $\F \subseteq \binom{U}{\leq d}$ be a family such that
 
\begin{itemize}
 \item[(1)] there are at least $l^{d-|C|}$ supersets~$F \supseteq C$ in $\F$ and
 \item[(2)] there are at most $l^{d-|C'|}$ supersets~$F' \supseteq C'$ in $\F$ for any other $C' \supsetneq C$, $C' \in \binom{U}{\leq d}$.
\end{itemize}
 
Then \F contains an $l$-flower with core~$C$ or $C \in \F$.

\end{lemma}

\begin{proof}
 Let us consider the restriction~$\F_C = \{ S \setminus C : S \in \F, S \supseteq C\}$ of sets in \F onto $C$. If $C \in \F$, then we are done. In the other case, let $X$ be a blocking set for $\F_C$, i.e., $X \cap F \neq \emptyset$ for all $F \in \F_C$ (by assumption $C \notin \F$, thus $\emptyset \notin \F_C$ and a blocking set exists). For every element $x \in X$ consider the number of sets in $\F_C$ that contain $x$; these correspond to the supersets of $C' = C \cup \{x\}$ in $\F$. We obtain from property (2) that there are at most $l^{d-|C'|} = l^{d-|C|-1}$ such sets. Thus, $|\F_C| \leq |X| \cdot l^{d-|C|-1}$ since every set in $\F_C$ has a non-empty intersection with $X$ while we have previously bounded the number of sets in $\F_C$ that contain at least one element of $X$. By property (1) of \F we have that $|\F_C| \geq l^{d-|C|}$. Therefore $|X| \geq l$ and consequently $\F' = \{S : S \in \F, S \supseteq C \}$ is the desired $l$-flower.
\end{proof}

The astute reader may find that Lemma~\ref{lem:flowercounting} is in a sense not completely tight. Indeed, if we require that there are more than $(l-1)^{d-|C|}$ supersets of $C$  (instead of at least that many) and at most $(l-1)^{d-|C'|}$ supersets for bigger cores $C'$, then an $l$-flower with core $C$ must exist (if $X$ hits $\F_C$, then $(l-1)^{d-|C|} < |\F_C| \leq |X| \cdot (l-1)^{d-|C|-1}$, therefore $|X| > l-1$). For technical convenience, the present formulation is more suitable for our algorithms.

\section{Logspace kernel for Hitting Set} \label{sec:logdhsk}
\introduceparameterizedproblem{\dHSk}{A set $U$ and a family $\F$ of subsets of $U$ each of size at most $d$, i.e., $\F \subseteq \binom{U}{\leq d}$, and $k \in \N$.}{$k$.}{Is there a $k$-hitting set $S$ for $\F$?}

In the following we present a logspace kernelization algorithm for \dHSk. The space requirement prevents the normal approach of finding sunflowers and modifying the family $\F$ in memory (we are basically left with the ability to have a constant amount of pointers and counters in memory). We start with an intuitive attempt for getting around the space restriction and show how it would fail.

The intuitive (but wrong) approach at a logspace kernelization works as follows. Process the sets $F \in \F$ one at a time and output $F$ unless we find that the subfamily of sets that where processed before $F$ contains a $(k+1)$-flower that enforces some core $C \subseteq F$ to be hit. For a single step $t$, let $\F_t$ be the sets that we have processed so far and let $\F'_t \subseteq \F_t$ be the family of sets in the output. We would like to maintain that a set $S$ is a $k$-hitting set for $\F_t$ if and only if it a $k$-hitting set for $\F'_t$. Now suppose that this holds and we want to show that our procedure preserves this property in step $t+1$ when some set $F$ is processed. This can only fail if we decide to discard $F$, and only in the sense that some $S$ is a $k$-hitting set for $\F'_{t+1}$ but not for $\F_{t+1}$ because $\F'_{t+1}\subseteq\F_{t+1}$. However, $S$ is also a $k$-hitting set for $\F'_t\subseteq \F'_{t+1}$ and, by assumption, also for $\F_t$. Recall that we have discarded $F$ because of a $(k+1)$-flower in $\F_t$ with core $C$, so $S$ must intersect $C$ (or fail to be a $k$-hitting set). Thus, $S$ intersects also $F\supseteq C$, making it a $k$-hitting set for $\F_t\cup\{F\}=\F_{t+1}$. 
Unfortunately, while the correctness proof would be this easy, such a straightforward approach fails as a consequence of the following lemma.\footnote{It is well known that finding a $k$-sunflower is \NP-hard in general. Similarly, finding a $k$-flower is \coNP-hard. For self-contained proofs see Appendix \ref{app:npflow}. Both proofs do not apply when the size of the set family exceeds the bounds in the (sun)flower lemma.}

\begin{lemma}
 Given a family of sets $\F$ of size $d$ and $F \in \F$. The problem of finding an $l$-flower in $\F \setminus F$ with core $C \subseteq F$ is \coNP-hard.
\end{lemma}

\begin{proof}
 We give a reduction from the \coNP-hard \notdHS problem which answers \yes for an instance if it does not have a hitting set of size at most $k$. Given an instance $(U, \F, k)$, we add a set $F$ that is disjoint from any set in $\F$ and ask if $\F$ contains a $(k+1)$-flower with core $C \subseteq F$. We show that this is the case if and only if $\F$ does not have a $k$-hitting set. Suppose that there is a $(k+1)$-flower with core $C \subseteq F$. By construction, $C = \emptyset$ since $F$ is completely disjoint from $\F$. Thus, $\F$ must contain a $(k+1)$-flower with an empty core, i.e., the hitting set size for $\F = \F_\emptyset$ is at least $k+1$. For the converse direction suppose that there is no hitting set of size $k$ for $\F$, i.e., a hitting set for $\F$ requires at least $k+1$ elements. Consequently, $\F$ is a $(k+1)$-flower with core $C = \emptyset \subseteq F$.
\end{proof}

Even if we know that the number of sets that where processed exceed the $k^d$ bound of Lemma~\ref{lem:flower}, finding out whether there is a flower with core $C \subseteq F$ is hard.\footnote{Note that we would run into the same obstacle if we use sunflowers instead of flowers; finding out whether there exists a $(k+1)$-sunflower with core $C \subseteq F$ for a specific set $F$ is \NP-hard as we show in Appendix~\ref{app:npsun}} Instead we use an application of Lemma~\ref{lem:flowercounting} that only ensures that there is some flower with $C \subseteq F$ if two stronger counting conditions are met. Whether condition (1) holds in $\F_t$ can be easily checked, but there is no guarantee that condition (2) holds if $\F_t$ exceeds a certain size bound, i.e., this does not give any guarantee on the size of the output if we process the input once. We fix this by taking a layered approach in which we filter out redundant sets for which there exist flowers that ensure that these sets must be hit, such that in each subsequent layer the size of the cores of these flowers decreases. 

We consider a collection of logspace algorithms $A_0,\ldots, A_d$ and families of sets $\F(0), \ldots, \F(d)$ where $\F(l)$ is the output of algorithm $A_l$. Each of these algorithms \emph{simulates} the next algorithm for decision-making, i.e., if we directly \emph{run} $A_l$, then it simulates $A_{l+1}$ which in turn simulates $A_{l+2}$, etc. If we run $A_l$ however, then it is the only algorithm that outputs sets; each of the algorithms that are being simulated as a result of running $A_l$ does not produce output.

We maintain the invariant that for all $C \subseteq U$ such that $l \leq |C| \leq d$, the family $\F(l)$ contains at most $(k+1)^{d-|C|}$ supersets of $C$. Each algorithm $A_l$ processes sets in $\F$ one at a time. For a single step $t$, let $\F_{t}$ be the sets that have been processed so far and let $\F_t(l)$ denote the sets in the output. Note that $A_l$ will not have access to its own output $\F_t(l)$ but we use it for analysis. 

Let us first describe how algorithm $A_d$ processes set $F$ in step $t+1$. If $F \notin \F_t$, then $A_d$ decides to output $F$; in the other case it proceeds with the next step. In other words $A_d$ is a simple algorithm that outputs a single copy of every set in $\F$. This ensures that the kernelization is robust for hitting set instances where multiple copies of a single set appear. If we are guaranteed that this is not the case, then simply outputting each set $F$ suffices. Clearly the invariant holds for $\F(d)$ since any $C \in \binom{U}{\geq d}$ only has $F = C$ as a superset and there is at most $(k+1)^0 = 1$ copy of each set in $\F(d)$. 

For $0 \leq l < d$ the procedure in Algorithm~\ref{alg:al} describes how $A_l$ processes $F$ in step $t+1$. First observe that lines 2 and 3 ensure that $\F(l) \subseteq \F(l+1)$. Assuming that the invariant holds for $\F(l+1), \ldots,\F(d)$, lines 8 and 9 ensure that the invariant is maintained for $\F(l)$. Crucially, we only need make sure that it additionally holds for $C \in \binom{U}{l}$ since larger cores are covered by the invariant for $\F(l+1)$.

\begin{algorithm} [t] \label{alg:al}
 simulate $A_{l+1}$ up to step $t+1$\;
 
 \eIf{$A_{l+1}$ decides not to output $F$}
  {
    do not output $F$ and end the computation for step $t+1$\;
  }
  {
    \ForEach{$C \subseteq F$ with $|C| = l$}
    {
      simulate $A_{l+1}$ up to step $t$\;
      count the number of supersets of $C$ that $A_{l+1}$ would output\;
      \If{the result is at least $(k+1)^{d-|C|}$}
      {
	do not output $F$ and end the computation for step $t+1$\;
      }
    }
    output $F$\;
  }
  \caption{Step $t+1$ of $A_l$, $0 \leq l < d$.}
\end{algorithm}

\begin{observation} \label{obs:logdhsk}
 $A_0$ commits at most $(k+1)^d$ sets to the output during the computation. This follows from the invariant for $\F(0)$ when considering $C = \emptyset$ which is in $\binom{U}{\geq l}$ for $l = 0$.
\end{observation}

\begin{lemma} \label{prop:logdhsk}
 For $0 \leq l \leq d$, $A_l$ can be implemented such that it uses logarithmic space and runs in $\Oh(|\F|^{d-l+2})$ time.
\end{lemma}

\begin{proof}
 Upon running $A_l$, at most $d-l$ algorithms (one instance of each $A_{l+1}, \ldots, A_d$) are being simulated at any given time, i.e., when we run $A_l$ at most $d-l+1$ algorithms actively require space for computation. In order to iterate over a family $\F$, a counter can be used to track progress, using $\log |\F|$~bits. Each algorithm can use such a counter to keep track of its current step. Let us assume that the elements in $U$ are represented as integers~$\{1, \ldots, |U|\}$. This enables us to iterate over sets~$C \subseteq F$ using a counter which takes $d \log |U|$~bits of space. Set comparison and verifying containment of a set~$C$ in a set in $\F$ can be done using constant space since the sets to be considered have cardinality at most $d$. Finally, each algorithm requires at most $\log ((k+1)^d) = d \log (k+1)$~bits to count these sets, where $k < |U|$ in any non-trivial instance. 
 
 Let us now analyze the running time. Clearly, $A_d$ runs in time $\Oh(|\F|^2$). There are at most $2^d$ subsets~$C \subseteq F$ for any set $F$ of size $d$. Thus, for $l \leq i < d$, each $A_i$ consults $A_{i+1}$ a total of $\Oh(|\F|)$ times during its computation (at most $1 + 2^d$ times in each step). All other operations take constant time, thus $A_l$ runs in time $\Oh(|\F|^{d-l}\cdot|\F|^2) = \Oh(|\F|^{d-l+2})$.
\end{proof}

Let us remark that we could also store each $C \subseteq F$ in line 5, allocate a counter for each of these sets, and simulate $A_{l+1}$ only once instead of starting a new simulation for each subset. This gives us a constant factor trade-off in running time versus space complexity. One might also consider a hybrid approach, e.g., by checking $x$ subsets of $F$ at a time. 

We will now proceed with a proof of correctness by showing that the answer to \dHSk is preserved in each layer $\F(0) \subseteq \F(1) \subseteq \ldots \subseteq \F(d)$.

\begin{lemma} \label{lem:logdhsk}
 Let $0 \leq l < d$ and let $S$ be a set of size at most $k$. It holds that $S$ is a hitting set for $\F(l)$ if and only if it is a hitting set for $\F(l+1)$.
\end{lemma}

\begin{proof}
 For each $0 \leq l < d$, we prove by induction over $0 \leq t \leq m$ that a set $S$ is a $k$-hitting set for $\F_t(l)$ if and only if it is a $k$-hitting set for $\F_t(l+1)$. This proves the lemma since $\F_m(l) = \F(l)$ and $\F_m(l+1) = \F(l+1)$. For $t = 0$ we have $\F_0(l) = \F_0(l+1) = \emptyset$ and the statement obviously holds. Let us assume that it holds for steps $t \leq i$ and prove that it also holds for step $i+1$. One direction is trivial: If there is a $k$-hitting set $S$ for $\F_{i+1}(l+1)$, then $S$ is also a hitting set for $\F_{i+1}(l) \subseteq \F_{i+1}(l+1)$. 
 
 For the converse direction let us suppose that a set $S$ is a $k$-hitting set for $\F_{i+1}(l)$. We must show that $S$ is also a hitting set for $\F_{i+1}(l+1)$. Let $F$ be the set that is processed in step $i+1$ of $A_l$. If $A_l$ decides to output $F$, then we have $\F_{i+1}(l) = \F_i(l) \cup \{F\}$, i.e., $S$ must hit $\F_i(l) \cup \{F\}$ and thus by the induction hypothesis it must also hit $\F_i(l+1) \cup \{F\} = \F_{i+1}(l+1)$.
 
 Now suppose that $A_l$ does not output $F$, i.e., $F \notin \F_{i+1}(l)$. First let us consider the easy case where $A_l$ decides not to output $F$ because $A_{l+1}$ decided not to output $F$. Thus, $F \notin \F_{i+1}(l+1)$ and $\F_{i+1}(l+1) = \F_{i}(l+1)$. A $k$-hitting set $S$ for $\F_{i+1}(l)$ hits $\F_i(l)$ and by the induction hypothesis it must hit $\F_i(l+1)$; we have that $S$ also hits $\F_{i+1}(l+1) = \F_{i}(l+1)$.
  
 In the other case we know that $A_l$ decides not to output $F$ because it established that there are at least $(k+1)^{d-|C|}$ supersets of some $C \subseteq F$ with $|C| = l$ in $\F_i(l+1)$. Furthermore, by the invariant for $\F(l+1)$ we have that for all sets $C'$ that are larger than $C$ there are at most $(k+1)^{d-|C|'}$ supersets $F' \subseteq C'$ in $\F_i(l+1) \subseteq \F(l+1)$. Consequently, by Lemma~\ref{lem:flowercounting} we have that $\F_i(l+1)$ contains $C$ or a $(k+1)$-flower with core $C$. Thus, in the first case any hitting set for $\F_i(l+1)$ must hit $C$, and in the second case any hitting set for $\F_i(l+1)$ requires at least $k+1$ elements if it avoids hitting $C$; in other words, a hitting set of size at most $k$ must hit $C$. Since $S$ is a $k$-hitting set for $\F_{i+1}(l)$ it must also hit $\F_{i}(l) \subseteq \F_{i+1}(l)$ and by the induction hypothesis we have that $S$ is also a $k$-hitting set for $\F_i(l+1)$. We have just established that $S$ must hit $C$ in order to hit $\F_i(l+1)$ since it has cardinality at most $k$. Thus, $S$ also hits $F \supseteq C$ and therefore $S$ is a hitting set for $\F_{i+1}(l+1) = \F_{i}(l+1) \cup \{F\}$.
\end{proof}
 
It is easy to see that a set $S$ is a $k$-hitting set for $\F(d)$ if and only if it is a hitting set for $\F$, because $A_d$ only discards duplicate sets. As a consequence of Lemma~\ref{lem:logdhsk}, a set $S$ is a $k$-hitting set for $\F(0)$ if and only if it is a $k$-hitting set for $\F(d)$. Therefore, it follows from Observation~\ref{obs:logdhsk} and Lemma~\ref{prop:logdhsk} that $A_0$ is a logspace kernelization algorithm for \dHSk.
 
\begin{theorem}
 \dHSk admits a logspace kernelization that runs in time $\Oh(|\F|^{d+2})$ and returns an equivalent instance with at most $(k+1)^d$ sets.
\end{theorem}

This kernelization is expressive; indeed, a subset of the input family is returned in the reduced instance and all minimal solutions up to size at most $k$ are preserved (the latter is consequence of any set $S$ of size $k$ being a hitting set for $\F(0)$ if and only if $S$ is a hitting set for $\F$). Let us remark that technically we still have to reduce the ground set to size polynomial in $k$. We can reduce the ground set of the output instance to at most $d(k+1)^d$ elements by including one more layer. Let $A_{d+1}$ be an algorithm that simulates $A_d$. Each time that $A_d$ decides to output a set $F$, algorithm $A_{d+1}$ determines the new identifier of each element $e$ in $F$ by counting the number of distinct elements that have been output by $A_d$ before the first occurrence of $e$. This can be done by simulating $A_d$ up to first step in which $A_d$ outputs $e$ by incrementing a counter each time an element is output for the first time (whether an element occurs for the first time can again be verified via simulation of $A_d$). We can take the same approach for the other logspace kernelizations given in this paper (either for ground sets or vertices).

\section{Logspace kernel for Set Packing} \label{sec:logdspk}
\introduceparameterizedproblem{\dSPk}{A set $U$ and a family $\F$ of subsets of $U$ each of size at most $d$, i.e., $\F \subseteq \binom{U}{\leq d}$, and $k \in \N$.}{$k$.}{Is there a $k$-packing $\Pee \subseteq \F$?}

In this section we present a logspace kernelization algorithm for \dSPk. The strategy for obtaining such a kernelization is similar to that in Section~\ref{sec:logdhsk}. However, the correctness proof gets more complicated. We point out the main differences.

We consider a collection of logspace algorithms $B_0, \ldots, B_d$ that perform almost the same steps as the collection of algorithms described in the logspace kernelization for \dHSk such that only the invariant differs. For each $0 \leq l \leq d$ we maintain that for all $C \subseteq U$ such that $l \leq |C| \leq d$, the family $\F(l)$ that is produced by $B_l$ contains at most $(d(k-1)+1)^{d-|C|}$ supersets of $C$. 

\begin{observation} \label{obs:logdsmk}
 $B_0$ commits at most $(d(k-1)+1)^d$ sets to the output during the computation. This follows from the invariant for $\F(0)$ when considering $C = \emptyset$ which is in $\binom{U}{\geq l}$ for $l = 0$.
\end{observation}

 Analogous to Lemma~\ref{prop:logdhsk} we obtain the following.

\begin{lemma} \label{prop:logdsmk}
 For $0 \leq l \leq d$, $B_l$ can be implemented such that it uses logarithmic space and runs in $\Oh(|\F|^{d-l+2})$ time.
\end{lemma}

The strategy for the proof of correctness is similar to that of Lemma~\ref{lem:logdhsk}. However, we need a slightly stronger induction hypothesis to account for the behavior of a solution for \dSPk since it is a subset of the considered family.

\begin{lemma} \label{lem:logdsmk}
 For $0 \leq l < d$, it holds that $\F(l)$ contains a packing $\Pee$ of size $k$ if and only if $\F(l+1)$ contains a packing $\Pee'$ of size at most $k$.
\end{lemma}
 
\begin{proof}
 For each $0 \leq l < d$, we prove by induction over $0 \leq t \leq m$ that for any $0 \leq j \leq k$ and any set~$S \subseteq \binom{U}{d(k-j)}$, $\F_t(l)$ contains a packing $\Pee$ of size $j$ such that $S$ does not intersect with any set in $\Pee$ if and only if $\F_t(l+1)$ contains a packing $\Pee'$ of size $j$ such that $S$ does not intersect with any set in $\Pee'$. This proves the lemma since $\F_m(l) = \F(l)$, $\F_m(l+1) = \F(l+1)$, and for packings of size $j = k$ the set $S$ is empty. It trivially holds for any $t$ if $j = 0$; hence we assume $0 < j \leq k$. For $t = 0$ we have $\F_0(l) = \F_0(l+1) = \emptyset$ and the statement obviously holds. Let us assume that it holds for steps $t \leq i$ and consider step $i+1$ in which $F$ is processed. If $\F_{i+1}(l)$ contains a packing $\Pee$ of size $j$, then $F_{i+1}(l+1)$ also contains $\Pee$ since $\F_{i+1}(l+1) \supseteq \F_{i+1}(l)$. Thus, the status for avoiding intersection with any set $S$ remains the same.
 
 For the converse direction let us assume that $\F_{i+1}(l+1)$ contains a packing $\Pee$ of size $j$ that avoids intersection with a set $S$ of size $d(k-j)$. We must show that $\F_{i+1}(l)$ also contains a $j$-packing that avoids $S$. Let $F$ be the set that is processed in step $i+1$ of $B_l$. Suppose that $F \notin \Pee$. Then $\F_{i}(l+1)$ already contains $\Pee$ and by the induction hypothesis we have that $\F_{i}(l) \subseteq \F_{i+1}(l)$ contains a $j$-packing $\Pee'$ that avoids $S$. 
 
 In the other case $F \in \Pee$. Suppose that $F \in \F_{i+1}(l)$, i.e., $B_l$ decided to output $F$. We know that $\F_{i}(l+1)$ contains the $(j-1)$-packing $\Pee \setminus \{F\}$ which avoids $S \cup F$. By the induction hypothesis we have that $\F_{i}(l)$ contains a packing $\Pee''$ of size $j-1$ that avoids $S \cup F$. Thus, $\F_{i+1}(l)$ contains the $j$-packing $\Pee' = \Pee'' \cup \{F\}$ which avoids $S$. 
 
 Now suppose that $\F \notin \F_{i+1}(l)$. By assumption, $F \in \Pee \subseteq \F_{i+1}(l+1)$. This implies that $B_l$ decided not to output $F$ because it has established that there are at least $d(k-1)+1)^{d-|C|}$ supersets of some $C \subseteq F$ with $|C| = l$ in $\F_i(l+1)$. Furthermore, by the invariant for $\F(l+1)$ we have that for all sets $C'$ that are larger than $C$ there are at most $(d(k-1)+1)^{d-|C'|}$ supersets $F' \supseteq C'$ in $F_i(l+1) \subseteq \F(l+1)$. Consequently, by Lemma~\ref{lem:flowercounting} we have that $\F_i(l+1)$ contains $C$ or a $(d(k-1)+1)$-flower with core $C$. In the first case, any hitting set for $\F_i(l+1)$ must hit $C$, and in the second case any hitting set for $\F_i(l+1)$ requires at least $d(k-1)+1$ elements if it avoids hitting $C$; thus any hitting set of size at most $d(k-1)$ must hit $C$. By assumption, $\Pee \setminus \{F\}$ and $S$ both avoid $C\subseteq F$ and therefore they both avoid at least one set $F'$ in $\F_i(l+1)$ since together they contain at most $d(k-1)$ elements. Thus we can obtain a $j$-packing $\Pee''$ in $\F_{i+1}(l+1)$ that also avoids $S$ by replacing $F$ with $F'$. Since $\Pee''$ no longer contains $F$ we find that $\F_i(l+1)$ contains $\Pee''$ and by the induction hypothesis there is some packing $\Pee'$ of size $j$ in $\F_i(l) \subseteq \F_{i+1}(l)$ that avoids $S$.
\end{proof}

Since $B_d$ only discards duplicate sets it holds that $\F(d)$ has a $k$-packing if and only if $\F$ has a $k$-packing. As a consequence of Lemma~\ref{lem:logdsmk} we have that $\F(0)$ has a $k$-packing if and only if $\F(d)$ has a $k$-packing. Therefore, it follows from Observation~\ref{obs:logdsmk} and Lemma~\ref{prop:logdsmk} that $B_0$ is a logspace kernelization algorithm for \dSPk.

\begin{theorem} \label{thm:logdsmk}
 \dSPk admits a logspace kernelization that runs in time $\Oh(|\F|^{d+2})$ and returns an equivalent instance with at most $(d(k-1)+1)^d$ sets.
\end{theorem}

\section{Logspace kernel for Edge Dominating Set} \label{sec:logedsk}
\introduceparameterizedproblem{\EDSk}{A graph $G(V,E)$ and $k \in \N$.}{$k$.}{Is there a set $S \subseteq E$ of at most $k$ edges such that every edge in $E \setminus S$ is incident with an edge in $S$?}

Our strategy for obtaining a logspace kernelization algorithm for \EDSk is as follows. First observe that the vertices of an edge dominating set of size at most $k$ are a vertex cover of size at most $2k$. This is frequently used in algorithms for this problem. Accordingly, we run our kernelization from Section~\ref{sec:logdhsk} for the case of vertex cover ($d=2$) with parameter $2k$, obtaining an equivalent instance $G'$ in which all minimal vertex covers of size at most $2k$ are preserved and proceed to add all edges between vertices of $G'$. 

We use some simulation to carry this out in logspace. Let \Rvc denote our logspace kernelization for \dHSk with $d = 2$ and parameter $2k$. The logspace kernelization algorithm \Reds for \EDSk proceeds as follows (see Algorithm~\ref{alg:edsk}). Count the number of vertices with degree at least $2k+1$ that \Rvc would output via simulation. If this is more than $2k$, then return a \no instance. Otherwise, output any edge between vertices that \Rvc would output, again via simulation. Let us now give an upper bound on the number of edges that \Reds will output.

\begin{algorithm} \label{alg:edsk} [t]
$c \leftarrow 0$\;
\ForEach{$v \in V$}
{
  simulate \Rvc\;
  \If{\Rvc finds a $(k+1)$-flower with empty core at any point during the computation}
  {return a \no instance\;}
  \If{\Rvc would output at least $2k+1$ edges incident to $v$}
  {$c \leftarrow c+1$\;}  
}
\eIf{$c > 2k$}
{ 
  return a \no instance\;
}
{
  \ForEach{$e = \{u, v\} \in E$}
  {
    simulate \Rvc\;
    \If{\Rvc would output at least one edge incident to $u$ and one edge incident to $v$}
    {
      output $e$ \;
    }
  }
} \caption{\Reds: Logspace kernel for \EDSk.}
\end{algorithm} 

\begin{lemma} \label{lem:edsize}
 \Reds commits $\Oh(k^3)$ edges to the output during computation.
\end{lemma}

\begin{proof}
First observe that \Rvc would output a set of edges $E''$ of size at most $(2k+1)^2 = 4k^2+2k$ (the number of edges that the logspace kernel for \dHSk with $d = 2$ and parameter $2k$ would output). Let $H \subseteq V(E'')$ denote vertices with degree at least $2k+1$, and let $L \subseteq V(E'')$ denote vertices with degree between $1$ and $2k$. Assume that \Reds does not output a \no instance. We have that $|H| \leq 2k$ and $|L| = \Oh(k^2)$. \Reds outputs all edges between vertices in $H$, all edges between vertices in $L$ and all edges between $H$ and $L$. Let us first bound the number of edges between vertices in $L$. If \Rvc would not output an edge $e$, then this is because it has determined that there is a $(k+1)$-flower with an endpoint of $e$ as its core (in other words, there are at least $k+1$ other edges incident with the same vertex). By assumption \Reds does not output a \no instance, therefore the case where an edge is discarded because there is a flower with empty core does not apply.  Thus, any edge that \Rvc would not output has at least one endpoint with degree $k+1$. Therefore, \Rvc would output all edges between vertices in $L$ and we have that there are at most $|E''| \leq 4k^2+2k$ edges between vertices in $L$. There are $\Oh(k^2)$ edges between vertices in $H$ and a further $\Oh(k \cdot k^2) = \Oh(k^3)$ edges between $H$ and $L$. Thus, \Reds outputs $\Oh(k^3)$ edges, as claimed.
\end{proof}

Let us now show that \Reds meets the time and space requirements for a logspace kernel.

\begin{lemma} \label{lem:edspace}
 \Reds can be implemented such that it uses logarithmic space and runs in $\Oh(|E|^4)$ time.
\end{lemma}

\begin{proof}
 Analogous to Lemma~\ref{prop:logdhsk} we have that \Rvc can be simulated in logspace and runs in time $\Oh(|E|^3)$ (we assume that $G$ is a simple graph; this saves a factor $|E|$ in the running time). It takes $\Oh(|V| \cdot |E|^3)$ time and logarithmic space to execute Lines 1 to 5: Besides the space reserved for simulating \Rvc, keep one counter to iterate over vertices in $V$ and another for counting the number of high degree vertices; both of these counters require at most $\log |V|$ bits. Executing Lines 6 to 7 clearly takes constant time and logspace. For Lines 8 to 12, reserve some memory to simulate \Rvc, two more bits in order to track if \Rvc outputs any edges incident to $u$ or $v$, and use $\log |E|$ bits for a counter to iterate over edges in $E$. This takes $\Oh(|E|\cdot |E|^3) = \Oh(|E|^4)$ time, i.e., the total running time is $\Oh(|V|\cdot |E|^3 + |E|^4) = \Oh(|E|^4)$.
\end{proof}

We proceed with a proof of correctness.

\begin{lemma} \label{lem:edcor}
 Let $G=(V, E)$ be the input graph for which \Reds outputs a graph $G'=(V',E')$. It holds that $G$ has an edge dominating set $S$ of size at most $k$ if and only if $G'$ has an edge dominating set $S'$ of size at most $k$.
\end{lemma}

\begin{proof}
 Suppose that $S$ is an edge dominating set of size at most $k$ for $G'$. Therefore, $V(S)$ is a vertex cover of size at most $2k$ for $G'$. Let $G''=(V',E'')$ be the subgraph of $G'=(V',E')$ that \Rvc would return. Now $V(S)$ is a vertex cover of size at most $2k$ for $G''$ and by Lemma~\ref{lem:logdhsk} we have that $V(S)$ is a $2k$-vertex cover for $G$. Therefore, $S$ is an edge dominating set for $G$ since the endpoints of edges in $S$ cover all edges in $E$.
 
 For the converse, suppose that $S$ is an edge dominating set of size at most $k$ for $G$. We know that $V(S)$ is a vertex cover of size at most $2k$ for $G$ and therefore also for $G'=(V', E')$ since $E' \subseteq E$. Now it could be the case that $S$ has some edges that only have one endpoint in $V'$. Thus we will obtain an edge dominating set $S'$ of size at most $k$ as follows, starting with $S' = \emptyset$. We consider each edge $e = \{u, v\} \in S$. If $u, v \in V'$, then $e \in E'$; we add $e$ to $S'$ in order to cover edges that are incident with $e$. Otherwise, w.l.o.g., we have $u \in V', v \notin V'$, i.e., there are no edges in $E'$ that are incident with $v$. Then $e$ can be replaced by any other edge $e'$ that is incident with $u$; we add $e'$ to $S'$.
\end{proof}

The following theorem is a consequence of Lemmas~\ref{lem:edsize} through~\ref{lem:edcor}.

\begin{theorem}
 \EDSk admits a logspace kernelization that runs in time $\Oh(|E|^4)$ and returns an equivalent instance with $\Oh(k^3)$ edges.
\end{theorem}

\section{Logspace kernelization for hitting and packing constant sized subgraphs} \label{sec:loggraph}
\introduceparameterizedproblem{\dHFVDk}{A graph $G(V,E)$, and $k \in \N$.}{$k$.}{Is there a set~$S$ of at most~$k$ vertices of $G$ such that $G[V\setminus S]$ does not have any $H \in \mathcal{H}$ as an induced subgraph?}

In this section we consider some hitting and packing problems on graphs. We will start with a logspace kernelization algorithm for \dHFVDk where we assume that for some constant~$d$ we have $|V(H)| \leq d$ for all $H \in \mathcal{H}$. We first describe a 'sloppy' version of this kernel for which we allow parallel edges to appear in the output graph. This algorithm, denoted as $R_0$, is an adaptation of the logspace kernelization for \dHSk in which we let $V$ play the role of the finite set $U$. The family $\F$ consists of those sets of vertices corresponding to the occurrences of graphs of $\mathcal{H}$ in $G$, i.e., a hitting set for $\F$ is a solution for \dHFVDk. 

Note that we must iterate over sets in $\F$ in order to run the \dHSk logspace kernelization while $\F$ is not given explicitly in the input. In order to do this the algorithm reserves $d \log n$ bits of space. This allows it to iterate over all subsets of vertices with cardinality at most $d$. It then identifies that such a set is in $\F$ if its induced graph coincides with a forbidden graph in $\mathcal{H}$, which can be verified in constant space. 

Only two such iterators are required, one for step by step processing of sets in $\F$, and another to obtain the count that is used to decide whether or not a set $F \in \F$ should appear in the output. When $R_0$ does decide to output $F$, we simply output all edges in $G[F]$. This is where parallel edges can appear since the sets of edges that are committed to the output may partially overlap. Using $R_0$ as a building block, we now present the proper logspace kernelization algorithm $R_1$ for \dHFVDk (Algorithm~\ref{alg:dhfvdk}).

\begin{algorithm} \label{alg:dhfvdk}.
 \ForEach{$e \in E$}
 {
  simulate $R_0$\;
  \If{$R_0$ decides to output $e$ for the first time}
  {
     output $e$\;
     halt the simulation of $R_0$\;
  }
}
\caption{$R_1$: Logspace kernel for \dHFVDk.}
\end{algorithm} 

\begin{observation}
 $R_1$ commits at most $\frac{d(d-1)}{2} \cdot (k+1)^d$ edges to the output during the computation. This follows from Observation~\ref{obs:logdhsk}, where each set corresponds to a graph with at most $d$ vertices and $\frac{d(d-1)}{2}$ edges.
\end{observation}

It is easy to see that $R_1$ resolves the issue with parallel edges and is executable in logspace. For running time analysis, note that family $\F$ has size $\Oh(\binom{|V|}{d})$. The running time for $R_0$ is $\Oh(\binom{|V|}{d}^{d+1})$ since we do not have multiple copies of sets in $\F$. Therefore $R_1$ runs in time $\Oh(|E| \cdot \binom{|V|}{d}^{d+1})$. Consequently, we give the following theorem.

\begin{theorem}
 \dHFVDk admits a logspace kernelization that runs in time $\Oh(|E|\cdot \binom{|V|}{d}^{d+1})$ and outputs an equivalent instance with at most $\frac{d(d-1)}{2} \cdot (k+1)^d$ edges, where $d$ is the maximum number of vertices of any $H \in \mathcal{H}$.
\end{theorem}

A similar adaptation of the \dSPk logspace kernelization algorithm yields a logspace kernel for the problem of finding a size $k$ disjoint union of occurrences of a constant sized graph $H$ in a host graph $G$, e.g., {\sc Triangle Packing($k$)}.

\begin{theorem}
 \dHPk admits a logspace kernelization that runs in time $\Oh(|E|\cdot \binom{|V|}{d}^{d+1})$ and outputs an equivalent instance with at most $\frac{d(d-1)}{2} \cdot (d(k-1)+1)^d$ edges, where $d = |V(H)|$.
\end{theorem}

\section{Linear-time kernel for Hitting Set} \label{sec:lindhsk}
We will now present a kernelization for \dHSk that runs in linear time. The algorithm processes sets in $\F$ one by one and decides whether they should appear in the final output or not. For a single step $t$, let $\F_t$ denote the sets that have been processed by the algorithm so far and let $\F'_t \subseteq \F_t$ be the sets stored in memory for which it has decided positively. The algorithm also uses a data structure in which it can store the number of supersets of a set $C \in \binom{U}{\leq d}$ that it has stored in $\F'_t$. Let supersets[$C$] denote the entry for set $C$ in this data structure. We begin by initializing supersets[$C$] $\leftarrow 0$ for each $C \subseteq F$ where $F \in \F$. 

We maintain the invariant that the number of sets $F \in \F'_t$ that contain $C \in \binom{U}{<d}$ is at most $(k+1)^{d-|C|}$. Now assume that the invariant holds after processing sets $\F_t \subseteq F$ and let $F$ be the next set. Algorithm~\ref{alg:lindhsk} describes how the algorithm processes $F$ in step $t+1$.

\begin{algorithm} \label{alg:lindhsk}
 \ForEach{$C \subseteq F$}
 {
  query supersets[$C$] in order to determine the number of supersets of $C$ in $\F'_t$\;
  \If{the result is at least $(k+1)^{d-|C|}$}
  {
    do not store $F$ and end the computation for step $t+1$\;
  }
 }
 store $F$\;
 \ForEach{$C \subseteq F$}
 {
   supersets[$C$] $\leftarrow$ supersets[$C$]+1\;
 }
 \caption{Step $t+1$ of the linear-time kernel for \dHSk.}
\end{algorithm}

Note that Line 3 and 4 ensure that the invariant is maintained. After all sets in $\F$ have been processed, the algorithm returns the family $\F' \subseteq \F$ of sets that it has stored in memory.

\begin{observation} \label{obs:lindhsk}
 After any step $t$, the algorithm has stored at most $(k+1)^d$ sets in $\F'_t$. This follows from the invariant when considering $C = \emptyset$.
\end{observation}

\begin{lemma} \label{prop:lindhsk}
 The algorithm can be implemented such that it runs in linear time.
\end{lemma}

\begin{proof}
 Before we run the algorithm, we sort the sets in $\F$ using the linear-time procedure described in \cite{van2014towards}. 
 Implementing the data structure as a trie \cite{van2014towards} allows us to query and update supersets[$C$] in time $\Oh(|C|)$ for a set $C$ where $|C| \leq d$. We perform both of these operations at most $2^d$ times in each step (this follows from the number of sets $C \subseteq F$ that we consider). Therefore we spend only constant time for each set that is processed and hence the algorithm runs in time~$\Oh(|\F|)$.
\end{proof}

We obtain the following lemma in a similar way to Lemma~\ref{lem:logdhsk}. The key difference is that we can now check if there are at least $(k+1)^{d-|C|}$ supersets of a set $C$ in $\F_t$ by using the data structure. Crucially, the algorithm has direct access to the previously processed sets that it plans to output; this simplifies the proof. 

\begin{lemma} \label{lem:lindhsk}
 When a family $\F_t \subseteq \F$ has been processed, the algorithm has a family $\F'_t \subseteq \F_t$ such that for any set $S$ of size at most $k$, $S$ is a hitting set for $\F_t$ if and only if $S$ is a hitting set for $\F'_t$.
\end{lemma}

\begin{proof}
 We prove the lemma by induction and show for each step $t$ that a set $S$ of size $k$ is a hitting set for $\F'_t$ if and only if $S$ is a hitting set for $\F_t$. For $t=0$ we have $\F'_0 = \F_0 = \emptyset$ and the statement obviously holds. Now let us assume that it holds for steps $t \leq i$ and prove that it also holds for step $i+1$. One direction is trivial: If a set $S$ is a $k$-hitting set for $\F_{i+1}$, then it is also a hitting set for $\F'_{i+1} \subseteq \F_{i+1}$. 
 
 For the converse direction let us suppose that a set $S$ is a $k$-hitting set for $\F'_{i+1}$. We must show that $S$ is also a hitting set for $\F_{i+1}$. Let $F$ be the set that is processed in step $i+1$. We must consider the case where the algorithm decides not to output $F$. Otherwise a hitting set for $\F'_{i+1}$ must hit $\F'_i \cup \{F\}$ and thus by the induction hypothesis it must also hit $\F_i \cup \{F\} = \F_{i+1}$.
 
 Now let us assume $F \notin \F'_{i+1}$. This implies that the algorithm decided not to output $F$ because for some $C \subseteq F$ it has determined that there are at least $(k+1)^{d-|C|}$ supersets of $F$ in $\F'_i$. Furthermore, by the invariant we have that there are at most $(k+1)^{d-|C'|}$ supersets of $C'$ in $\F'_i$ for any $C'$ that is larger than $C$. Consequently, by Lemma~\ref{lem:flowercounting} we have that $\F'_i$ contains $C$ or a $(k+1)$-flower with core $C$, i.e., any hitting set of size $k$ for $\F'_i$ must hit $C$ and therefore also hits $F \supseteq C$. Let $S$ be a $k$-hitting set for $\F'_{i+1}$ and note that $S$ is also a $k$-hitting set for $\F'_i \subseteq \F'_{i+1}$. By the induction hypothesis, $S$ is also a hitting set for $\F_i$. We have just established that $S$ must hit $C$ in order to hit $\F'_i$. Thus, $S$ hits $F$ and is a hitting set for $\F_{i+1} = \F_i \cup \{F\}$.\end{proof}

Reducing the size of the ground set is simple since we have the equivalent instance in memory. In one pass we compute a mapping of elements occurring in the equivalent instance to $d(k+1)^d$ distinct identifiers. The following theorem is a consequence of Observation~\ref{obs:lindhsk}, Lemma~\ref{prop:lindhsk} and Lemma~\ref{lem:lindhsk}.

\begin{theorem} \label{thm:lindhsk}
 \dHSk admits a linear-time kernelization which returns an equivalent instance with at most $(k+1)^d$ sets.
\end{theorem}

\section{Linear-time kernel for Set Packing} \label{sec:lindspk}
In the following we will present a linear-time kernelization for \dSPk. The algorithm performs almost the same steps as the linear-time kernelization for \dHSk such that only the invariant differs. In this case we maintain that the number of sets $F \in \F'_t$ that contain $C \in \binom{U}{\leq d}$ is at most $(d(k-1)+1)^{d-|C|}$. We point out the main differences.

\begin{observation} \label{obs:lindspk}
 After any step $t$, the algorithm has stored at most $(d(k-1)+1)^d$ sets in $\F'_t$. This follows from the invariant when considering $C = \emptyset$.
\end{observation}

Analogous to Lemma~\ref{prop:lindhsk} we obtain the following.

\begin{lemma} \label{prop:lindspk}
 The algorithm can be implemented such that it runs in linear time.
\end{lemma}

Let us proceed with a proof of correctness.

\begin{lemma} \label{lem:lindspk}
 When a family $\F_t \subseteq \F$ has been processed, the algorithm has a set $\F'_t\subseteq \F_t$ such that $\F'_t$ contains a packing $\Pee$ of size $k$ if and only if $\F_t$ contains a packing $\Pee'$ of size $k$. 
\end{lemma}

\begin{proof}
 We prove the lemma by induction and show for each step $t$ that for $0 \leq j \leq k$ and any set $S \subseteq \binom{U}{d(k-j)}$, $\F'_t$ contains a packing $\Pee$ of size $j$ such that $S$ does not intersect with any set in $\Pee$ if and only if $\F_t$ contains a packing $\Pee'$ of size $j$ such that $S$ does not intersect with any set in $\Pee$. This implies the lemma statement since for packings of size $j=k$ the set $S$ is empty. It trivially holds for any $t$ if $j=0$; hence we assume that $0 < j \leq k$. For $t = 0$ we have $\F'_0 = \F_0 = \emptyset$ and the statement obviously holds. Now let us assume that it holds for steps $t \leq i$ and prove that it also holds for step $i+1$.  If $\F'_{i+1}$ contains a packing $\Pee$ of size $j$, then $\F_{i+1}$ also contains $\Pee$ since $\F_{i+1} \supseteq \F'_{i+1}$. Thus, the status for avoiding intersection with any set $S$ remains the same. 
 
 For the converse direction let us assume that $\F_{i+1}$ contains a packing $\Pee$ of size $j$ that avoids intersection with a set $S$ of size $d(k-j)$. We must show that $\F'_{i+1}$ also contains a $j$-packing that avoids $S$. Let $F$ be the set that is processed in step $i+1$. Suppose that $F \notin \Pee$. Then $\F_i$ contains $\Pee$ and by the induction hypothesis we have that $\F'_i \subseteq \F'_{i+1}$ contains a packing $\Pee'$ of size $j$ that avoids $S$. 
 
 In the other case $F \in \Pee$. Suppose that $F \in \F'_{i+1}$. We know that $\F_i$ contains the packing $\Pee \setminus \{F\}$ which avoids $S \cup F$. By the induction hypothesis we have that $\F'_i$ contains a packing $\Pee''$ of size $j-1$ that avoids $S \cup F$. Hence, $\F'_{i+1}$ contains the $j$-packing $\Pee' = \Pee'' \cup \{F\}$ which avoids $S$. 
 
 Now suppose that $F \notin \F'_{i+1}$. This implies that the algorithm decided not to output $F$ because it has determined that there are at least $(d(k-1)+1)^{d-|C|}$ supersets of some set $C \subseteq F$ in $\F'_i$. Furthermore, by the invariant we have that there are at most $(d(k-1)+1)^{d-|C'|}$ supersets of $C'$ in $\F'_i$ for any $C'$ that is larger than $C$. Consequently, by Lemma~\ref{lem:flowercounting} we have that $\F'_i$ contains $C$ or a $(d(k-1)+1)$-flower with core $C$, i.e., any hitting set for $\F'_i$ of size at most $d(k-1)$ for $\F'_i$ must hit $C$. We show that there is a $j$-packing in $\F'_{i+1}$ that avoids $S$: We know that $\Pee \setminus \{F\}$ is a $(j-1)$-packing in $\F_i$ that avoids $S \cup F$. By the induction hypothesis, there is a $(j-1)$-packing $\Pee''$ in $\F'_i$ that also avoids $S \cup F$.  By assumption, $\Pee''$ and $S$ both avoid $C \subseteq F$ and therefore they also avoid at least one set $F'$ in $\F'_i$ since together they contain at most $d(k-1)$ elements. Thus, we can obtain the packing $\Pee' = \Pee'' \cup \{F\}$ of size $j$ which avoids~$S$.
\end{proof}

The following theorem is a consequence of Observation~\ref{obs:lindspk}, Lemma~\ref{prop:lindspk} and Lemma~\ref{lem:lindspk}.

\begin{theorem}
 \dSPk admits a linear-time kernelization which returns an equivalent instance with at most $(d(k-1)+1)^d$ sets.
\end{theorem}

\section{Concluding remarks} \label{sec:conc}
In this paper we have presented logspace kernelization algorithms for \dHSk and \dSPk. We have shown how these can be used to obtain logspace kernels for hitting and packing problems on graphs, and have given a logspace kernelization for \EDSk.
By using flowers instead of sunflowers we save a large hidden constant in the size of these kernels. Furthermore, we have improved upon a linear-time kernel for \dHSk and have given a linear-time kernel for \dSPk. One question to settle would be whether a vertex-linear kernel for \VCk can be found in logspace. While known procedures for obtaining a vertex-linear kernel seem unsuitable for adaptation to logspace, we currently also do not have the tools for ruling out such a kernel. The problem of finding a vertex-linear kernel for \VCk in linear time also remains open.

\bibliographystyle{abbrv}
\bibliography{references}

\newpage
\appendix

\section{Finding a maximum (sun)flowers} \label{app:npflow}
\subsection{\NP-hardness proof for finding a sunflower of size $l$.}
\begin{proof}
 Reduction from \dSP, $d \geq 3$. Suppose $(U, \F, k)$ is an instance of \dSP. We obtain an instance $(U', \F', l)$ for \textsc{sunflower} as follows. Let $m = |\F|$. Make disjoint copies $\F_1, \ldots, \F_{m+1}$ of $\F$ such that $\F_i$ has subsets over finite set $U_i$. Let $U' = U_1 \cup \ldots \cup U_{m+1}$, let $\F' = \F_1 \cup \ldots \cup \F_{m+1}$, and let $l = k(m+1)$. We show that $\F$ has a set packing $\Pee$ of size $k$ if and only if $\F'$ contains a sunflower of with $l$ petals. First, suppose $\F$ has a packing $\Pee$ of size $k$. Then we obtain a sunflower with $k(m+1)$ petals and an empty core by taking $k$ sets corresponding to $\Pee$ in each $\F_i$. For the converse direction, suppose that $\F'$ has a sunflower with $k(m+1)$ petals and core $C$. Since each family $\F_i$ contains at most $m$ sets we know that there must be at least one petal in $\F_i$ and another in $\F_j$ for some $i \neq j$. Thus, $C = \emptyset$. Furthermore, since there are $k(m+1)$ petals, there must be some $i$ such that at least $k$ petals are in $\F_i$. These $k$ petals form a packing since they must be completely disjoint. 
\end{proof}

\subsection{\coNP-hardness proof for finding a flower of size $l$.}
\begin{proof}
 Reduction from \notdHS.  Suppose $(U, \F, k)$ is an instance of \notdHS. We obtain an instance $(U', \F', l)$ for $k$-\textsc{flower} as follows. Let $m = |\F|$. Make disjoint copies $\F_1, \ldots, \F_{m+1}$ of $\F$ such that $\F_i$ has subsets over finite set $U_i$. Let $U' = U_1 \cup \ldots \cup U_{m+1}$, let $\F' = \F_1 \cup \ldots \cup \F_{m+1}$, and let $l = (k+1)(m+1)$. We show that $\F$ does not have a hitting set $S$ of size $k$ if and only if $\F'$ contains a $l$-flower. First, suppose $\F$ does not have a $k$-hitting set $S$. Therefore, there must be some collection of sets in $\F$ that requires at least $k+1$ elements to hit. By taking copies of these sets from each $\F_i$ and choosing $C = \emptyset$ we obtain a $(k+1)(m+1)$-flower. For the converse direction, suppose that $\F'$ has a $(k+1)(m+1)$-flower with core $C$. Since each family $\F_i$ contains at most $m$ sets we know that this $(k+1)(m+1)$-flower must contain some sets from $\F_i$ and $\F_j$ such that $i \neq j$. Thus, $C = \emptyset$. Furthermore, since by definition we have that $(k+1)(m+1)$ elements are required to hit $\F'$, we know that there must be some $\F_i$ where we need at least $k+1$ elements for a hitting set. 
\end{proof}

\subsection{\NP-hardness proof for finding a sunflower with $l$ petals and core $C \subseteq F$ for a specific set $F$.} \label{app:npsun}

\begin{proof}
 Reduction from \dSP, $d \geq 3$. Given an instance $(U, \F, k)$, we add a set $F$ that is disjoint from any set in $\F$ and ask if $\F$ contains a sunflower with $k$ petals and core $C \subseteq F$. We show that this is the case if and only if $\F$ has a set packing of size $k$. Suppose that $\F$ has a sunflower with $k$ petals and core $C \subseteq F$. By construction, $C = \emptyset$ since $F$ is completely disjoint from $\F$, i.e., the sunflower is a set packing of size $k$. For the converse direction suppose that $\F$ contains a $k$-set packing $\Pee$. Thus, $\Pee$ is also a sunflower with core $C = \emptyset \subseteq F$. 
\end{proof}

\end{document}